\documentclass[11pt]{article}

\usepackage{fullpage}
\usepackage{epic,eepic}
\usepackage{amsmath}
\usepackage{amssymb}
\usepackage[dvips]{graphicx}
\usepackage[dvips]{color}
\usepackage{chngpage}

\def\Ot{\tilde{O}}
\def\Lovasz{Lov{\'a}sz }

\newtheorem{theorem}{Theorem}
\newtheorem{lemma}[theorem]{Lemma}
\newtheorem{claim}[theorem]{Claim}
\newtheorem{corollary}[theorem]{Corollary}

\newtheorem{definition}[theorem]{Definition}
\newtheorem{property}{Property}

\def\cal{\mathcal}

\newcommand{\lca}{\mbox{lca}}

\newcommand{\poly}{\mbox{poly}}

\newcommand{\eps}{\varepsilon}

\newcommand{\ignore}[1]{}

\newcommand{\etal}{{\em et al.\ }}

\newenvironment{proof}{{\bf Proof:\ }}{\hfill$\Box$\medskip}

\begin{document}
\title{Fast, precise and dynamic distance queries}
\author{
\and Yair Bartal\thanks{
 Hebrew University.
 \texttt{\scriptsize yair@cs.huji.ac.il}
}
\and Lee-Ad Gottlieb\thanks{
 Weizmann Institute of Science.
 \texttt{\scriptsize lee-ad.gottlieb@weizmann.ac.il}
}
\and Tsvi Kopelowitz\thanks{
 Bar Ilan University.
 \texttt{\scriptsize \{kopelot,moshe,liamr\}@cs.biu.ac.il}
}
\and Moshe Lewenstein\footnotemark[3]
\and Liam Roditty\footnotemark[3]
}

\date{}
\maketitle
\thispagestyle{empty}

\begin{abstract}\noindent
We present an approximate distance oracle for a point set $S$ with $n$ points 
and doubling dimension $\lambda$. For every $\eps>0$, the oracle supports 
$(1+\eps)$-approximate distance queries in (universal) constant time, occupies 
space $[\eps^{-O(\lambda)} + 2^{O(\lambda \log \lambda)}]n$, and can be 
constructed in $[2^{O(\lambda)}\log^3 n + \eps^{-O(\lambda)} + 2^{O(\lambda 
\log \lambda)}]n$ expected time. This improves upon the best previously known 
constructions, presented by Har-Peled and Mendel~\cite{HaPe06}. Furthermore, 
the oracle can be made fully dynamic with expected $O(1)$ query time and only 
$2^{O(\lambda)} \log n + \eps^{-O(\lambda)} + 2^{O(\lambda \log \lambda)}$ 
update time. This is the first fully dynamic $(1+\eps)$-distance oracle.

\end{abstract}

\newpage
\setcounter{page}{1}

\section{Introduction}
A distance oracle for a set of $n$ points $S$ under some distance 
function $d(\cdot,\cdot)$, is a preprocessed data structure that given two 
points $x,y \in S$ returns their distance without needing to query the 
distance function. Distance oracles are of interest when the distance 
function is too large to store (for example when the function is a distance 
matrix storing all $O(n^2)$ interpoint distances) or when querying the 
distance function is expensive (for example when the distance function is 
defined by graph-induced distances).

Distance oracles were introduced in a seminal paper of Thorup and Zwick
\cite{ThZw05}. For a weighted undirected graph, they gave an $(2k-1)$-approximate 
distance oracle with query time $O(k)$, for $k \ge 1$. Immediate from these
runtimes is the question of reduced query time,
and in fact Mendel and Naor~\cite{MeNa06} recently presented a 
$O(k)$-approximate distance oracle for general metrics with $O(1)$ query time.
Another direction for improvement is in the approximation guarantee, and 
distance oracles with $(1+\eps)$-approximation ($0<\eps \le \frac{1}{2}$) 
have been achieved for planar graphs \cite{Klein02,Thorup04},
geometric graphs \cite{GuLeNaSm08} and doubling spaces \cite{HaPe06}.

A further question in this field is that of dynamic distance oracles. In this
setting, the point set $S$ is updated with the removal or addition of points,
and the distance oracle must be updated accordingly. A similar paradigm was considered
by \cite{RZ04}, who gave a distance oracle for an unweighted undirected graph under the 
removal of edges. Here, the distance function is the shortest path metric of the 
underlying graph, which must be consulted during an oracle update.

In this paper we consider a metric space with doubling dimension $\lambda$,
and present a $(1+\eps)$-approximate distance that answers
queries in (universal) constant time. The distance oracle occupies near-optimal space
$[\eps^{-O(\lambda)} + 2^{O(\lambda \log \lambda)}]n$, and
can be constructed in 
$[2^{O(\lambda)}\log^3 n + \eps^{-O(\lambda)} + 2^{O(\lambda \log \lambda)}]n$  expected  time.
This improves upon the best previously known constructions in this setting, presented by 
Har-Peled and Mendel~\cite{HaPe06}.
Furthermore, this oracle can be made fully dynamic with expected $O(1)$ query time. In this case, 
the update time is only
$2^{O(\lambda)} \log n + \eps^{-O(\lambda)} + 2^{O(\lambda \log \lambda)}$
time per point.

\paragraph{Related work.}
Thorup and Zwick \cite{ThZw05} demonstrated that a weighted undirected graph can be
preprocessed to create an oracle that can answer $(2k-1)$-approximate
distance queries between any two vertices in $O(k)$ time. The structure is
of size $O(n^{1+1/k})$, and the randomized preprocessing takes
$O(mn^{1/k})$ time (where $n$ is the number of vertices and $m$ is the
number of edges). Roditty, Thorup and Zwick~\cite{RoThZw05} gave a
deterministic preprocessing algorithm that builds the distance oracle in
$\tilde{O}(mn^{1/k})$ time. Baswana and Sen~\cite{BaSe06} and Baswana and
Kavitha~\cite{BaKa06} improved the deterministic preprocessing time to
$\Ot(\min(m \sqrt{n}, kn^{2+1/k}))$. Mendel and Naor~\cite{MeNa06} showed
that for any metric space there exists an $O(k)$-approximate distance
oracle of size $O(n^{1 + 1/k})$ that supports queries in constant time
independent of $k$.

Turning to lower bounds, Thorup and Zwick~\cite{ThZw05} proved that any
$(2k+1)$-approximate distance oracle must have size at least
$\min(m,\Omega(n^{1+1/k}))$. Very recently, Sommer, Verbin, and
Yu~\cite{SoVeYu09} extended a technique of P\v{a}tra\c{s}cu~\cite{Pa08} to
prove that a $k$-approximate distance oracle preprocessed in $t$ time must
occupy $n^{1+\Omega(1/tk)}$ space.

While the previous results apply to arbitrary metric space, distance
oracles have also been studied for more restricted settings.
Klein~\cite{Klein02} and Thorup~\cite{Thorup04} considered planar graphs,
and showed how to build a $(1+\eps)$-distance oracle with $O((n \log
n)/\eps)$ space and $O(\eps^{-1})$ query time. (Thorup~\cite{Thorup04}
presented an oracle for directed planar graphs.) Gudmundsson, Levcopoulos,
Narasimhan and Smid~\cite{GuLeNaSm08} considered geometric (Euclidean)
graphs that are $t$-spanners for some constant $t>1$. (A graph $G = (S,E)$
is said to be a $t$-spanner for $S$, if for any point pair $p,q \in S$,
there exists in $G$ a path connecting $p$ and $q$, and the length of this
path is at most $t$ times the true distance between $p$ and $q$.) They
showed how to construct a $(1+\eps)$-approximate distance oracle of size
$O((\frac{t}{\eps})^d n \log n)$. Their oracle can be constructed in
$(\frac{t}{\eps(t-1)})^{O(d)} n \log n$ time, and answers distance queries
in $O((\frac{t}{\eps})^d)$ time. Har-Peled and Mendel~\cite{HaPe06} considered  metric spaces with low doubling dimension. They presented two data structures both of size $\eps^{-O(\lambda)}n$ (which attains the lower-bound on the space required for this task).
Their first data structure can be constructed in $2^{O(\lambda)} n \log ^2 n$ time and answers $(1+\eps)$-approximate distance queries in $2^{O(\lambda)}$ time. Their second data structure can be constructed in polynomial time and answers $(1+\eps)$-approximate distance queries in $O(\lambda)$ time.

\paragraph{Our contribution.}
Our result improves the query time from $O(\lambda)$ in the construction of Har-Peled and Mendel~\cite{HaPe06} to constant time,
while also providing the first fully dynamic oracle construction.
As in Har-Peled and Mendel \cite{HaPe06}, an immediate application of our static oracle for $S$
is a $(1+\eps)$-approximate distance oracle for every graph which is a
$t$-spanner of $S$. Our static oracle is a dramatic improvement over those 
of Gudmundsson \etal \cite{GuLeNaSm08} reducing the query time to constant in several aspects, 
while the space is smaller by a factor of $\log n$ and the setting is more general. 

To obtain our improved bounds, we present contributions in several distinct areas, including 
dynamic embeddings and dynamic tree structures. We present two probabilistic dynamic 
embeddings for doubling spaces: The first is into a tree metric, and the second is a 
snowflake embedding into Euclidean space (see Section~\ref{sec:dyn-embed}). In both cases, we 
are interested in the probability of low distortion, as opposed to expectation. This seems to 
be the first consideration of dynamic embeddings (although the related notion of on-line 
embeddings recently appeared in \cite{InMa10}). We also present a powerful dynamic tree 
structure that allows a binary search over centroid paths in our setting (see 
Section~\ref{sec:backup1}). Our oracle framework and the tools used in this paper further 
imply other distance oracles with various tradeoffs. A brief summary of these results is 
found in Table~\ref{tab:misc}.

\renewcommand{\arraystretch}{1.2}
\begin{table}
\small
\centering
\begin{tabular}{||c|c|c|c|c||}
\hline
 & Reference & Static construction & Query time & Space  \\
\hline\hline 
1& \cite{HaPe06} & [$2^{O(\lambda)} \log ^2 n + \eps^{-O(\lambda)}]n$ & $2^{O(\lambda)}$ 
& $\eps^{-O(\lambda)}n$ \\
\hline
2& \cite{HaPe06} & $\poly(n)$ & $O(\lambda)$ & $\eps^{-O(\lambda)}n$ \\
\hline
3& Section~\ref{sec:backup1} & $[2^{O(\lambda)}\log n + \eps^{-O(\lambda)}]n$ & 
$O(\log\log n)$ & $\eps^{-O(\lambda)}n$ \\
\hline
4& Section~\ref{sec:misc} & $[2^{O(\lambda)}\log^3 n + \eps^{-O(\lambda)}]n$* & 
$O(\log\log\log n)$ & $\eps^{-O(\lambda)}n$ \\
\hline
5& Section~\ref{sec:misc} & $[2^{O(\lambda)}\log^3 n + \eps^{-O(\lambda)} + 2^{O(\lambda 
\log \lambda)}]n$* & $O(\log\log \lambda)$ & $\eps^{-O(\lambda)}n$ \\
\hline
6& Theorem~\ref{thm:static-oracle} & $[2^{O(\lambda)}\log^3 n + \eps^{-O(\lambda)} + 
2^{O(\lambda \log \lambda)}]n$* & $O(1)$ & $[\eps^{-O(\lambda)}+ 2^{O(\lambda \log \lambda)}]n$ \\
\hline\hline
 & Reference & Dynamic updates & Query time & Space  \\
\hline\hline 
7& Theorem~\ref{thm:backup2} & $2^{O(\lambda)} \log n + \eps^{-O(\lambda)}$ & 
$2^{O(\lambda)}$ & $\eps^{-O(\lambda)}n$ \\
\hline
8& Theorem~\ref{thm:backup1} & $2^{O(\lambda)} \log n + \eps^{-O(\lambda)}$ & $O(\log^2 
\log n)$ & $\eps^{-O(\lambda)}n$ \\
\hline
9& Section~\ref{sec:misc} & $2^{O(\lambda)} \log n + \eps^{-O(\lambda)} + 2^{O(\lambda 
\log \lambda)}$ & $O(\min \{\log \log \lambda, \log \log \log \log n \})*$ & 
$\eps^{-O(\lambda)}n$ \\
\hline
10& Theorem~\ref{thm:dynamic-oracle} & $2^{O(\lambda)} \log n + \eps^{-O(\lambda)} + 
2^{O(\lambda \log \lambda)}$ & $O(1)*$ & $[\eps^{-O(\lambda)}+ 2^{O(\lambda \log \lambda)}]n$ \\
\hline
\end{tabular}
\caption{A summary of $(1+\eps)$ distance oracles. *In expectation.}\label{tab:misc}
\end{table}
\normalsize

\vspace{-10pt}
\paragraph{Paper outline.}
The rest of this paper is organized as follows. We first present preliminary points
in the next section. We then describe (in Section~\ref{sec:backbone}) the basic 
structure that forms the backbone of most of our constructions. We proceed to present 
the central contribution of this paper, the $O(1)$ query time oracles 
(both static and dynamic) in Section~\ref{sec:main}. The dynamic oracle requires two 
dynamic backup oracles, which are separate constructions of independent interest
(presented in Section~\ref{sec:backup}), and both oracles require several technical 
contributions (presented in Section~\ref{sec:technical}).

\section{Preliminaries}\label{sec:prelim}

Here we review some preliminary definitions and results that are
required in order to present our new ideas.

\paragraph{\bf Lowest common ancestor query.}
A lowest common ancestor (LCA) query on tree $T$ provides two nodes $u,v$ of $T$,
and asks for the node $w$ that is ancestral to $u$ and $v$, and is minimal in the
sense that no descendant of $w$ is ancestral to both $u$ and $v$.
LCA queries can be answered in $O(1)$ time in the word
RAM model, using a linear size data structure \cite{CoHa05}.

In the dynamic setting, Cole and Hariharan \cite{CoHa05} gave
a linear size data structure that supports LCA queries under insertions and
deletions of leaves and internal nodes to the tree. The query and update times are
all $O(1)$ under the word RAM model. We can extend their structure to also identify in
$O(1)$ time the two children of $w$ that are ancestors of $x$ and $y$
(Lemma \ref{lem:lca-child}).

\paragraph{\bf Level ancestor query.} A level ancestor query on tree $T$
provides a node $u$ and level $k$, and asks for the node $w$ that is both an
ancestor of $u$ and is $k$ nodes removed from the root of $T$. There exists a
linear size structure that supports level ancestor queries in $O(1)$ time.

In the dynamic setting, there exists a structure that supports level ancestor
queries in $O(1)$ search and update time under insertions of leaves into $T$.
However, the insertion of internal nodes is not supported by this structure
\cite{AlHo00,KoLe07}. For the purposes of this paper, we must maintain a tree under the
insertions of internals nodes, hence we are unable to utilize standard level
ancestor query structures in our dynamic setting.

\ignore{
\paragraph{\bf Membership query.}
Given a set of numbers $N$, there exists a linear data structure that
preprocesses the data in linear time, and can determine in $O(1)$ time
whether a given number $q$ is a member of the set $N$~\cite{FrKoSz84}.
For our purposes, we will not require a dynamic version of this structure.

\paragraph{Predecessor query.}
The predecessor query problem is defined as follows: Given a set of
numbers $S \subset U$ of size $n$ over a universe $U$ of size $N$,
preprocess $S$ so that given a query number $q \in U$, the largest number
in $S$ less than $q$ can be efficiently computed. There exists a linear
space data structure that can be constructed in $O(n \log n)$ time and
which answers predecessor queries in time $O(\log\log N)$ in the word RAM
model \cite{}.

For our purposes, we will not require a dynamic version of this structure.
}

\paragraph{Doubling dimension.}

For a metric $(X,d)$, let $\lambda$ be the smallest value such that every
ball in $X$ can be covered by $2^\lambda$ balls of half the radius. The {\em
doubling dimension} of $X$ is dim$(X)=\lambda$. A metric is {\em doubling}
when its doubling dimension is constant. Note that while low Euclidean
dimension implies low doubling dimension (Euclidean metrics of dimension
$d$ have doubling dimension $O(d)$ \cite{GuKrLe03}), low doubling
dimension is more general than low Euclidean dimension.
The following property can be demonstrated via a repetitive application of
the doubling property.

\begin{property}[Packing property]\label{prop:doubling}
For set $S$ with doubling dimension $\lambda$, if the minimum
interpoint distance in $S$ is at least $a$, and the diameter
of $S$ is at most $b$, then $|S| \le 2^{O(\lambda\log
(b/a))}$.
\end{property}

\paragraph{Hierarchical Partitions.}\label{sec-hier}

Similar to what was described in \cite{GaoGuiNgu04,KrLe04}, a subset of
points $X \subseteq Y$ is an $(r,s)$-discrete center set (or {\em net} in
the terminology of \cite{KrLe04}) of $Y$ ($r \le s$) if it satisfies the
following properties:
\renewcommand{\labelenumi}{(\roman{enumi})}
\begin{enumerate}
\item Packing: For every $x,y \in X$, $d(x,y) \ge r$.
\item Covering: Every point $y \in Y$ is strictly within distance $s$ of
some point $x \in X$: $d(x,y) < s$.
\end{enumerate}
The previous conditions require that the points of $X$ be spaced out, yet
nevertheless cover all points of $Y$.

A hierarchical partition for a set $S$ is a hierarchy of discrete center
sets, where each level of the hierarchy is a discrete center set of the
level beneath it. Krauthgamer and Lee \cite{KrLe04} gave a fully dynamic
hierarchy that can be updated in time $2^{O(\lambda)}\log \alpha$
($\alpha$ is the aspect ratio of $S$), where a single update to $S$ can result in
$2^{O(\lambda)}\log \alpha$ updates to the hierarchy. Cole and
Gottlieb \cite{CoGo06} presented a semi-dynamic hierarchy, where a single insertion into
$S$ can result in the insertion of $2^{O(\lambda)}$ points into the hierarchy.
However, points cannot be removed from within the hierarchy, and after many deletions the
hierarchy is rebuilt in the background.\footnote{It suffices, if the hierarchy holds $n'$
nodes (included those nodes storing deleted points), to start
rebuilding after $\frac{n'}{3}$ deletions, and to complete the
rebuilding over the next $\frac{n'}{6}$ insertions and deletions;
that is, for each update to the point set 7 updates are performed
on the background structure. The completed hierarchy
will then contain at least $\frac{n'}{2}$ points, including at
most $\frac{n'}{6}$ deleted points.}

Our constructions can make use of either hierarchy, but our descriptions will
assume the hierarchy of \cite{CoGo06}.
The bottom level of this hierarchy is the set $Y_1 =
Y_{5^0} = S$ that contains all points, and the top level $Y_{5^{\lceil
\log_5 \alpha \rceil}}$ contains only a single point.
(For ease of presentation, we assume throughout
that the minimum interpoint distance in $S$ is 1.) Each intermediate
level $0 < i < \lceil \log_5 \alpha \rceil$ is represented by a set $Y_{5^i}$, which is a
$(\frac{1}{5}5^i,\frac{3}{5}5^i)$-discrete center set for $Y_{5^{i-1}}$.
The {\em radius} of level $Y_{5^m}$ is defined to be $5^m$. A point $y \in
Y_{5^m}$ $b$-covers a point $x \in Y_{5^l}$ if $d(y,x) < b \cdot 5^l$,
and the covering property states that each points in the hierarchy is
$\frac{3}{5}$-covered by some point one level up. It can be shown by a
repeated application of the covering property that each point is
$\frac{4}{5}$-covered by some point in every higher level.

The hierarchy may be augmented with neighbor links: Each point $x \in
Y_{5^m}$ records what points of $Y_{5^m}$ are within distance $b \cdot
5^m$ of $x$ -- these are the $b$-neighbors of $x$. By the packing property, a point may have
$b^{O(\lambda)}$ such neighbors. To save space in the hierarchy, points
that have no $b$-neighbors (where $b\ge 2$) and also cover only one point
in the next level may be represented implicitly. This compression scheme
ensures that a hierarchy and neighbor links can be stored with
$b^{O(\lambda)}n$ space.

\paragraph{\bf Snowflake embedding.}
Assouad's \cite{Assouad83} snowflake embedding -- as improved by Gupta \etal
\cite{GuKrLe03} -- takes an arbitrary metric $(X,d)$ of
doubling dimension $\lambda$ (where $d$ is the metric's distance
function), and embeds the {\em snowflake} $(X,d^{1/2})$ into $O(\lambda
\log \lambda)$-dimensional Euclidean space with $O(\lambda)$ distortion.
That is, the embedding into Euclidean space achieves low
dimension and distortion, but has the `side effect' that every
interpoint distance in the metric is replaced by its square root, with
distortion to the square root at most $O(\lambda)$. Har-Peled and Mendel
\cite{HaPe06} used this embedding in the context of distance oracles.

Although \cite{GuKrLe03} did not give an exact run time
for their static embedding, the following analysis holds: The static
embedding can be achieved by first building a point hierarchy that records
$O(\lambda)$-neighbors; this can be done in
$[2^{O(\lambda)} \min \{\log n,\log \alpha \} + 2^{O(\lambda \log \lambda)}]n$ time
\cite{KrLe04,HaPe06,CoGo06}. The analysis in \cite{GuKrLe03} requires a
constructive application of the \Lovasz Local Lemma \cite{MoTa10}, which in this case
can be done in $2^{O(\lambda \log \lambda)}n$ expected time.
Given these constructions, the image for each point can easily be computed in
$2^{O(\lambda \log \lambda)}n \log \alpha$ time, but a more careful analysis
shows that $2^{O(\lambda \log \lambda)}$ work per
hierarchical point is sufficient. It follows that the entire construction can be done in
$[2^{O(\lambda)} \min \{\log n,\log \alpha \} + 2^{O(\lambda \log \lambda)}]n$
expected time. The construction of \cite{GuKrLe03} is static, and so for our purposes we
will need to create a dynamic version of the embedding (see
Section~\ref{sec:dyn-snow}).

\section{Construction backbone}\label{sec:backbone}

The backbone of our distance oracles is a point hierarchy, and we shall employ the semi-dynamic 
hierarchy of Cole and Gottlieb \cite{CoGo06} augmented with storage of $c$-neighbor pairs, and 
the distance between the neighbors in a pair. We will show below that $c=\frac{8}{5\eps}$ 
($0<\eps\le \frac{1}{2}$) is an appropriate choice. On top of this hierarchy, we define a 
parent-child relationship as follows: every point $x \in Y_{5^i}$ is a child of some point $y 
\in Y_{5^{i+1}}$ that $\frac{3}{5}$-covers $x$. This implies that points that are siblings must 
be $6$-neighbors. The parent-child relationship immediately 
defines an ancestor-descendant relationship as well. (Note that some other constructions in 
this paper, such as the dynamic embeddings of Section~\ref{sec:dyn-embed}, will require a 
different definition of the parent-child relationship.) We have the following property:

\begin{property}[Hereditary property]\label{prop:hereditary}
If two points $x,y \in Y_{5^i}$ are $c$-neighbors
then their respective parents $x',y' \in Y_{5^{i+1}}$ are $c$-neighbors as well.
\end{property}

That is, if two points $x,y \in Y_{5^i}$ have the property that
$d(x,y)\le c 5^i$, then their respective parents $x',y' \in Y_{5^{i+1}}$ have the property that
$d(x',y') \le d(x',x) + d(x,y) + d(y,y')
< \frac{3}{5}5^{i+1} + c5^i + \frac{3}{5}5^{i+1}
= \frac{6}{5}5^{i+1} + c5^i
= (\frac{6+c}{5})5^{i+1}
< c5^{i+1}$. 
This means that $x,y$ and all their ancestors up to their lowest common ancestor are all
found explicitly in the hierarchy.

The hierarchical tree $T$ is extracted from the hierarchy. $T$ has one node per 
hierarchical point, and the points of hierarchical level $Y_{5^i}$ all have 
corresponding nodes in tree level $T_i$. The parent-child relationship among points in 
the hierarchy defines the same parent-child relationship among the corresponding tree 
nodes. We will refer to the distance between two tree nodes, by which we mean the 
distance between their corresponding points. Further, we compress all nodes whose 
points are only implicit in the hierarchy; this results in the contraction of some 
unary paths. This tree will allow us to navigate the hierarchy.

\paragraph{Structural lemmas.} Here we present lemmas that will be used to prove 
correctness of our oracles. The key observation motivating our constructions is 
captured by the following lemma, a variant of which was central for the construction 
of low stretch spanners \cite{GaoGuiNgu04,Ro07,GoRo08a,GoRo08b}. While we state the 
lemmas in term of general $b$, we are actually interested in the two cases where $b=6$ 
and $b=c=\frac{8}{5\eps}$.

\begin{theorem}\label{thm:approx}
Let $x,y \in Y_{5^l}$ be a pair that are not $b$-neighbors, and let $x',y' \in Y_{5^m}$ in some
level $m>l$ be the lowest respective ancestors of $x$ and $y$ that are $b$-neighbors. Then
$d(x',y')$ is a $(1 \pm \frac{8}{5b})$-approximation to $d(x,y)$.
\end{theorem}

\begin{proof}
The fact that $x$ and $y$ are not $b$-neighbors implies that $d(x,y) > b5^l$, while the fact
that $x'$ and $y'$ are $b$-neighbors implies that $d(x',y') \le b5^m$. The parent-child
relationship implies that
$d(x,x'), d(y,y') \le \sum_{i=l+1}^m \frac{3}{5} \cdot 5^m < \frac{4}{5} \cdot 5^m$.
It follows that
$d(x,y)
\le d(x,x') + d(x',y') + d(y',y)
< d(x',y') + \frac{8}{5} 5^m
\le d(x',y') + \frac{8}{5b} d(x',y')
= (1+\frac{8}{5b}) d(x',y')
$.
Also,
$d(x,y)
\ge d(x',y') -d(x,x')- d(y',y)
< d(x',y') - \frac{8}{5} 5^m
< d(x',y') - \frac{8}{5b} d(x',y')
= (1-\frac{8}{5b}) d(x',y')
$.
\end{proof}

When $b=6$ we have that $d(x',y')$ is a $(1 \pm \frac{4}{15})$-approximation for 
$d(x,y)$, and when $b=c$ we have that $d(x',y')$ is a $(1 \pm \eps)$-approximation for 
$d(x,y)$. Hence, the problem of finding a $(1+\eps)$-approximation for $d(x,y)$ can be 
solved by finding the lowest ancestral $c$-neighbors of $x$ and $y$ in the hierarchy. 
Later, we will also make use of the following corollary.

\begin{corollary}\label{cor:approx}
A $\delta$-approximation for $d(x',y')$ implies a $\delta(1 \pm \frac{8}{5b})$-approximation for
$d(x,y)$, and vice versa.
\end{corollary}

We have proved that finding the lowest ancestral $c$-neighbors of $x$ and $y$ in the
hierarchy will provide a $(1+\eps)$-approximation to $d(x,y)$. The following lemma demonstrates
a close relationship between the value of $d(x,y)$, and the level in which 
the lowest ancestral $b$-neighbors $x',y'$ of $x,y$ are found.

\begin{lemma}\label{lem:level}
Let $i$ be such that  $5^{i-1} < d(x,y) \le 5^i$. Let $x'$ and $y'$, be the
respective ancestors of $x$ and $y$ in level $Y_{5^p}$.
\renewcommand{\labelenumi}{(\roman{enumi})}
\begin{enumerate}
\item
If $p < i - 1 - \log_5(b+\frac{8}{5})$, then $x'$ and $y'$ are not $b$-neighbors.
\item
If $p \ge i - \log_5(b-\frac{8}{5})$, $x'$ and $y'$ must be either $b$-neighbors or the same
point.
\end{enumerate}
\end{lemma}

\begin{proof}
\noindent (i) We have that
$d(x',y')
\ge d(x,y) - d(x',x) -d(y,y')
> 5^{i-1} - 2\frac{4}{5} \cdot 5^p
= 5^{i-1} - \frac{8}{5} \cdot 5^p$.
Note that for values
$p < i - 1 - \log_5(b+\frac{8}{5})$,
we have that
$d(x',y')
> (b+\frac{8}{5})5^p - \frac{8}{5} \cdot 5^p
= b \cdot 5^p$,
and $x''$ and $y''$ cannot be $b$-neighbors.

\smallskip
\noindent (ii) We have that
$d(x',y')
\le d(x,y) + d(x',x) + d(y,y')
< 5^i + \frac{8}{5} \cdot 5^p$.
Note that for values
$p \ge i - \log_5(b-\frac{8}{5})$,
we have that
$d(x',y')
\le (b-\frac{8}{5})5^p + \frac{8}{5} \cdot 5^p
= b \cdot 5^p$,
and $x'$ and $y'$ must be $b$-neighbors or the same point.
\end{proof}

Lemma~\ref{lem:level} implies that the level of the lowest ancestral
$b$-neighbors $x',y'$ is in the
range $[i - 1 - \log_5(b+\frac{8}{5}),i - \log_5(b-\frac{8}{5})]$,
and this range is of size $O(1)$ irrespective of the value of $b$.
Crucially, this means that the levels of the lowest ancestral $6$-neighbors
of $x,y$ and the lowest ancestral $c$-neighbors of $x,y$ differ by a fixed
value ($\log c$), up to an additive constant. This implies that
finding the lowest common $6$-neighbors of $x,y$
is a useful tool to find their lowest common $c$-neighbors, and therefore
a $(1+\eps)$-approximation to $d(x,y)$. A further consequence of
Lemma~\ref{lem:level} is that a $\delta$-approximation for $d(x,y)$ (or in
fact for any descendants of $x',y'$) is
sufficient to pinpoint the level of the lowest ancestral $b$-neighbors
$x',y'$ to a range of $\log \delta + O(1)$ possible levels.

\paragraph{Deletions.} In closing this section, we note that for all dynamic 
structures
presented in this paper, deletions are handled by rebuilding in the background
(as was described in Section~\ref{sec:prelim} in the context of dynamic hierarchies):
Deleted points are kept in the structure, and when a large number of points have been deleted,
we begin to rebuild the structure in the background. This has no effect on the asymptotic
runtimes of our constructions.

\section{Oracle queries in $O(1)$ time}\label{sec:main}
In this section we present $(1+\epsilon)$-approximate distance oracles with $O(1)$
query time and size $[\eps^{-O(\lambda)} + 2^{O(\lambda \log \lambda)}]n$. The
first oracle we present is static, and the second is a dynamic version of the
static construction. In Section~\ref{sec:misc}, we briefly discuss variants of these 
constructions that appear in Table~\ref{tab:misc}

\subsection{Static oracle}\label{sec:main-static}

In this section we prove the following theorem:
\begin{theorem}\label{thm:static-oracle}
There exists a static $(1+\epsilon)$ approximate distance oracles with $O(1)$
query time and size $[\eps^{-O(\lambda)} + 2^{O(\lambda \log \lambda)}]n$. The
oracle can be updated in expected time
$[2^{O(\lambda)}\log^3 n + \eps^{-O(\lambda)} + 2^{O(\lambda \log \lambda)}]n$.
\end{theorem}

Given points $x$ and $y$, the oracle finds in $O(1)$ time the lowest ancestral
$6$-neighbors $x',y'$ of $x,y$. As a consequence of Lemma~\ref{lem:level}, the
level of the lowest ancestral $6$-neighbors of $x,y$ gives the level of the
lowest ancestral $c$-neighbors of $x,y$ to within an additive constant. 
The level of the $c$-neighbors can then be found using a constant number of level
ancestor queries.

The oracle locates $x',y'$ in three steps, each of which can be implemented in $O(1)$
time: In the first step we compute an $O(\log n)$ approximation to $d(x,y)$.
As a consequence of Lemma \ref{lem:level}, this approximation restricts the
candidate level of $x',y'$ to a range of $O(\log\log n)$ possible levels.
The second step then derives an $O(\lambda^3)$ approximation to $d(x,y)$,
which further restricts the candidate level to $O(\log \lambda)$ possible
levels. The third step locates $x',y'$.

\paragraph{Step 1.}
The first step provides an $O(\log n)$-approximation for $d(x,y)$, which implies a $O(\log n)$
approximation for $d(x',y')$ (by Corollary~\ref{cor:approx}). First note that for
doubling metrics, there exists a $6$-stretch spanner with $2^{O(\lambda)}n$ edges that can
be constructed in time $2^{O(\lambda)}n \log n$ \cite{GaoGuiNgu04, GoRo08b}. Given this
spanner, we can construct the oracle of Mendel and Schwab \cite[Theorem 2(2)]{MeSc09} with
parameter $k=O(\log n)$, which yields an $O(\log n)$-approximate distance oracle of size
$O(n)$ that supports distance queries in $O(1)$ time, with expected construction time
$2^{O(\lambda)}n \log^3 n$. We construct the oracle in the preprocessing stage, and derive
an $O(\log n)$-approximation for $d(x,y)$ -- and therefore for $d(x',y')$ -- in $O(1)$ time.

\paragraph{Step 2.} 
The second step gives an $O(\lambda^3)$-approximation to $d(x',y')$, assuming an $O(\log 
n)$-approximation is already known. If $\lambda \ge \log^{1/3} n$, this step is unnecessary and is 
skipped. We therefore assume that $\lambda < \log^{1/3} n$.

Recall that the $O(\log n)$ approximation to $d(x',y')$ restricts the candidate levels of
$x',y'$ to a range of $r=O(\log\log n)$ levels (Lemma~\ref{lem:level}). The ancestors of
$x$ and $y$ in the top level of this range can be located via a level ancestor query on $x$
and the desired number of levels below $\lca(x,y)$ (assuming that we have recorded for every
node in $T$ its distance from the root). But the task of locating the ancestors of $x,y$ in
the bottom level of this range is frustrated by the fact that some ancestors of $x,y$ below 
$x',y'$ may be
compressed (if these nodes are below the lowest ancestral $c$-neighbors of $x,y$); these
uncompressed nodes will be ignored by the level ancestor query, which will therefore return an
incorrect level. To solve this problem, we preprocess a $\log \log n$-jump tree for $T$ (see
Section~\ref{sec:jump-tree}). A series of $\log \log n$-jump queries locate in $O(1)$ time
explicit ancestors of $x,y$ that are at most $\log\log n$ levels below $x',y'$, which will
suffice for our purposes. Call these ancestors $x'',y''$ -- By Corollary~\ref{cor:approx}, an
$O(\lambda^3)$-approximation to $d(x'',y'')$ yields an $O(\lambda^3)$-approximation to 
$d(x',y')$.

Now, for every node $u \in T_{i \cdot \log\log n}$, $i \ge 0$, (including implicit nodes) let
the neighbor set $N_u$ contain all explicit nodes that are descendants of $u$ and $u$'s
$6$-neighbors, in $r+2\log\log n$ levels below $T_i$. We preprocess the snowflake embedding
for each non-empty neighborhood, which can be done in total time $[2^{O(\lambda)}\log n +
2^{O(\lambda \log \lambda)}]n$ (since each explicit node participates in $2^{O(\lambda)}$ 
neighborhoods). Since $\lambda < \log^{1/3} n$, the target dimension of the
snowflake embedding is $d = \log^{1/3} n \log\log n$. Since the aspect ratio of each
neighborhood is $O(\log n)$ and the embedding has distortion $O(\lambda)$,
each coordinate can be stored in $b = O(\log \lambda+\log\log n) = O(\log\log n)$ bits.
Therefore $b^2d = o(\log n)$. It follows from Lemma \ref{lem:bit-dist} (see 
Section~\ref{sec:bit-trick}) that each vector
may be stored in $O(1)$ words, and the embedding distance between two vectors returned in
$O(1)$ time. Squaring the embedding distance gives a $O(\lambda^2)$ approximation to the true
distance.

It remains only to locate a neighborhood containing both $x''$ and $y''$, for which it
suffices to locate $x''$'s ancestor in the lowest level $T_{i \cdot \log\log n}$ above the
candidate range. A pointer to this ancestor can be preprocessed in time $O(\log\log n)$ per
node. Given the correct neighborhood, a $O(\lambda^{2})$-approximation for $d(x'',y'')$ can be
found in $O(1)$ time, and this yields a $O(\lambda^{2})=O(\lambda^{3})$-approximation for 
$d(x',y')$.

\paragraph{Step 3.}
The third step locates $x',y'$ in $O(1)$ time,
under the assumption that a $O(\lambda^{3})$-approximation to $d(x',y')$ is known.

As in Step 2, the $O(\lambda^{3})$ approximation to $d(x',y')$ restricts the candidate
levels of $x',y'$ to a range of $r= 3\log_5 \lambda +O(1)$ levels. The top level of this range 
is found using a level ancestor query, and then a constant number of $\log_5 \lambda$-jump
queries locate explicit ancestors of $x,y$ that are at most $\log_5 \lambda$ levels below
$x',y'$. Call these ancestors $x'',y''$, and let their level (or the level of the lower one)
be $T_i$. Note that 
$d(x'',y'') 
\le d(x'',y') + d(y',x') + d(x',y'')
< \frac{4}{5} 5^{i+4\log_5 \lambda+O(1)} 
+ 6 \cdot 5^{i+4\log_5 \lambda+O(1)} 
+ \frac{4}{5} 5^{i+4\log_5 \lambda+O(1)}
= \frac{38}{5} 5^{i+4\log_5 \lambda+O(1)}
= O(\lambda^{4} 5^i)$.

In the preprocessing stage, we find for each explicit node $u \in T_i$ all nodes in 
levels $T_{i-r-\log\lambda}$ through $T_i$ whose distance to $u$ is $O(\lambda^{4}5^i)$. 
For each such pair, we preprocess their lowest ancestral $6$-neighbors, which can all be 
done in time $2^{O(\lambda \log \lambda)}$ time per point. Now, given $x''$ and $y''$, 
their lowest ancestral $6$-neighbors can be located in $O(1)$ time.

\subsection{Dynamic oracle}\label{sec:main-dynamic}

In this section we give a dynamic version of the oracle. We prove the following theorem:
\begin{theorem}\label{thm:dynamic-oracle}
There exists a dynamic $(1+\epsilon)$-approximate distance oracle with expected $O(1)$ and
worst-case
$\min\{2^{O(\lambda)},O(\log^2\log n)\}$
query time, and size
$[\eps^{-O(\lambda)} + 2^{O(\lambda \log \lambda)}]n$.
The oracle can be maintained dynamically in 
$2^{O(\lambda)}\log n + \eps^{-O(\lambda)} + 2^{O(\lambda \log \lambda)}$
time per point update.
\end{theorem}

The dynamic oracle is given points $x$ and $y$ as a query, and runs the two backup oracles of
Section~\ref{sec:backup} in the background. Between them, these oracles locate the
lowest ancestral $c$-neighbors of $x,y$ in
$p=\min \{2^{O(\lambda)},O(\log^2 \log n) \}$
time (Theorems~\ref{thm:backup1} and \ref{thm:backup2}).

The oracle itself searches for the lowest ancestral $6$-neighbors $x',y'$ of $x,y$. After
locating these nodes, a $\log c$-jump query can be used to descend in the tree to within a
constant number of levels of the lowest ancestral $c$-neighbors of $x,y$. As before, The
oracle locates $x',y'$ in three steps, each of which can be implemented in $O(1)$ time: In
the first step we use the probabilistic dynamic tree embedding of Section~\ref{sec:dyn-tree}
to compute a $p^{O(1)}$-approximation to $d(x,y)$ and therefore to $d(x',y')$. In the second
step we use the probabilistic snowflake embedding of Section~\ref{sec:dyn-snow} to compute
an $O(\lambda^2)$-approximation to $d(x',y')$. In the third step we locate $x',y'$.

\paragraph{Step 1.}
The first step provides a $p^{O(1)}$-approximation to $d(x,y)$, which implies a $p^{O(1)}$
approximation to $d(x',y')$. We utilize the dynamic tree embedding of
Lemma~\ref{lem:tree-embed} with parameter $i=\log_{5\lambda/4} p$. The probability that the
embedding fails to give the desired distortion $p^{O(1)}$ is given as $O(1/p)$. Since the
backup oracles
run in time $O(p)$, the event of failure does not affect the target expected runtime of
$O(1)$.

\paragraph{Step 2.}
The second step gives a $O(\lambda^2)$-approximation to $d(x',y')$, assuming an
$p^{O(1)}$-approximation is already known. If $\lambda \ge p^{O(1)}$, this step is skipped.
We therefore assume that $\lambda < p^{O(1)}$.

Recall that the $p^{O(1)}$ approximation to $d(x',y')$ restricts the candidate levels of
$x',y'$ to a range of $r=O(\log p)$ levels (Lemma~\ref{lem:level}). The top level in this range
is provided by an LCA query on the dynamic tree of Section~\ref{sec:dyn-tree}.
Then a series of $\log p$-jump queries locate in $O(1)$ time
explicit ancestors of $x,y$ that are at most $\log p$ levels below $x',y'$, which will
suffice for our purposes. Call these ancestors $x'',y''$ -- an $O(\lambda^2)$-approximation to
$d(x'',y'')$ yields an $O(\lambda^2)$-approximation to $d(x,y)$.

Similar to what was done before, we preprocess the dynamic snowflake embedding of
Section~\ref{sec:dyn-snow} for each non-empty neighborhood. Our target dimension for the
embedding is $d = O(\log p)$, so it follows from Theorem~\ref{thm:snowflake} that the
embedding achieves an $O(\lambda)$-approximation with probability of failure only $O(1/p)$,
which does not affect the expected $O(1)$ runtime of the oracle. Since the aspect ratio of
each neighborhood is $O(p)$ and the embedding has distortion $O(\lambda) = O(p)$, each
coordinate can be stored in $b = O(\log p)$ bits. Therefore $b^2d = o(\log n)$, and it
follows from Lemma \ref{lem:bit-dist} that each vector may be stored in $O(1)$ words, and
the embedding distance between two vectors returned in $O(1)$ time. Squaring the embedding
distance gives a $O(\lambda^2)$ approximation to the true distance.

We then locate a neighborhood containing both $x''$ and $y''$, for which it suffices to
locate $x''$'s ancestor in the lowest level $T_{i \cdot \log p}$ above the candidate range.
A pointer to this ancestor can be recorded dynamically in time $O(\log p)$ per node
insertion into $T$. Given the correct neighborhood, a $O(\lambda^{2})$-approximation to
$d(x'',y'')$ can be found in $O(1)$ time, and this yields a $O(\lambda^{2})$-approximation
to $d(x',y')$.

\paragraph{Step 3.}
The third step provides a constant factor approximation to $d(x,y)$ in $O(1)$ time,
under the assumption that we are provided a $O(\lambda^{2})$ approximation to $d(x,y)$.
The $O(\lambda^2)$ approximation to $d(x',y')$ restricts the candidate
levels of $x',y'$ to a range of $r=O(\log \lambda)$ levels. We can ascend to the bottom
level of this range via pointers from $x''$ and $y''$, and these pointers can be maintained
dynamically in $O(\log p)$ time per insertion into $T$.
The rest of the construction for this step is identical to the third step of the static
oracle, and can be done in $2^{O(\lambda \log \lambda)}$ time and space per node insertion
into $T$.

\subsection{Variant constructions.}\label{sec:misc}
Here, we briefly discuss three variant constructions that appear in 
Table~\ref{tab:misc}. We show that these constructions can find the lowest ancestral 
$6$-neighbors of $x,y$, after which the lowest ancestral $c$-neighbors of $x,y$ can be 
found easily by using a level ancestor query or a $k$-jump query.

Construction 4 is achieved by first running the $O(\log n)$-approximate 
oracle of \cite{MeSc09}. As mentioned in Section~\ref{sec:main-static}, this 
approximation restrics the candidate levels of the lowest ancestral $6$-neighbors to 
$O(\log\log n)$ levels. Using level ancestor queries, a binary search finds the correct 
level in $O(\log\log\log n)$ query time.

Construction 5 is achieved by running the static construction until the end of Step 2.
At the end of Step 2, the range of candidate levels in $O(\log \lambda)$, and a binary 
search on this range finds the correct level in $O(\log\log \lambda)$ query time.

Construction 9 runs the dynamic construction until the end of Step 2, at which
point the range of candidate levels is reduced to 
$O(\min \{\log \lambda, \log\log\log n \})$. 
A binary search on these levels can be executed using at most $\log\log\lambda$ 
different $k$-jump trees, resulting in query time
$O(\min \{\log\log \lambda, \log\log\log\log n \})$.

\section{Backup oracles}\label{sec:backup}

In this section, we present two dynamic oracles that find the lowest ancestral $c$-neighbors of 
points $x,y$. The maintenance of both oracles is bounded by the time to maintain a hierarchy. The 
first oracle answers query in time $2^{O(\lambda)}$, and the second in time $O(\log^2 \log n)$. 
While we have presented these constructions as backup oracles, it should be noted that they are 
contributions of independent interest.

\subsection{Dynamic oracle queries in $2^{O(\lambda)}$ time}\label{sec:backup2}

In this section, we give a dynamic oracle that given $x$ and $y$, finds their lowest 
ancestral $c$-neighbors in the hierarchy, thereby deriving a 
$(1+\eps)$-approximation to $d(x,y)$. We prove the following theorem:

\begin{theorem}\label{thm:backup2}
There exists a dynamic oracle that given $x,y \in S$ returns a 
$(1+\eps)$-approximation to
$d(x,y)$ in $2^{O(\lambda)}$ time, and supports updates in time
$2^{O(\lambda)}\log n + \eps^{-O(\lambda)}$.
\end{theorem}

An overview of the construction is as follows. Given hierarchy tree $T$, we
create a forest of $2^{O(\lambda)}$ distinct trees. The difference
between these trees lies solely in their parent-child relationship.
We then show that in at least one of these trees, $x$ and $y$ have their lowest
common ancestor at level $\log_5 d(x,y)$, or within $O(1)$ levels of this
level. By Lemma~\ref{lem:level} this level is within $O(1)$ levels of the
lowest ancestral $6$-neighbors of $x,y$.

\paragraph{Construction.}
We create a forest of distinct trees $\cal{T}=\{T^1,\ldots,
T^{\ell}\}$ in a manner similar to the creation of $T$.
Each tree is built on top of the point hierarchy, so all
trees share the same nodes and tree level sets. However,
we ignore every odd level of the hierarchy, so the
trees of $\cal{T}$ only have non-odd levels. Each point in
$Y_{5^j}$ (for non-odd $j$) corresponds to a unique node in tree
level $j$ of each tree $T^h$.

It remains to describe the parent-child assignments for the trees of
$\cal{T}$. A node $u \in T^h_j$ is assigned a single parent
$v \in T^h_{j+2}$ which covers $u$. Crucially, ties among
candidate parents are not broken arbitrarily (as they were for
tree $T$). Rather, each tree $T^h \in \cal{T}$ possesses a distinct
set of {\em dominant} nodes in each tree level. Given a group of
candidate parent nodes, the dominant node in the group always
takes the child. We stipulate that the distance between dominant
nodes $T^h_j$ must be greater than $2\cdot 5^j$, so that two
dominant nodes cannot vie for the same child. Further, we stipulate
that a node in $T_j$ is dominant in exactly one tree of $\cal{T}$.
Clearly, a forest of size $|\cal{T}| = 2^{\Theta(\lambda)}$ can
obey these stipulations.

The dominance assignment can be implemented as follows: When a
point $x$ is added to hierarchical level $Y_{5^j}$, a corresponding
node $u$ is added to $T^h_j$ for each $T^h \in \cal{T}$. In one
of these trees, $u$ is chosen to be dominant. (Note that by the
packing property of doubling spaces, there must be at least one
tree in which $u$ is not within distance $2 \cdot 5^j$ of any other
dominant node in the same level.) In each tree $T^h \in \cal{T}$,
$u$ is assigned as a child of the dominant node in $T^h_{j+2}$ that
covers $u$, or of an arbitrary node of $T^h_{j+2}$ if there is no dominant one.
Note that once a parent-child assignment is made, the assignment
cannot be reversed. Hence, a newly inserted dominant node does
not become the parent of previously inserted nodes that it covers.
(A reassignment would necessitate a cut-link operation on the tree,
which is not supported by either \cite{CoGo06} or \cite{CoHa05}.)
The entire forest $\cal{T}$ can be maintain in time $2^{O(\lambda)}$ per node
insertion into $T$.
The distance between $u \in T^h_j$ and its ancestor $w \in T^h_m$ is less than
$5^m \sum_{i=0}^\infty 25^{-i}
= \frac{5^m}{1-1/25}
= \frac{25}{24}\cdot 5^m$.

\paragraph{Oracle query}
Let $x$ and $y$ be two points such that
$5^{i-1} < d(x,y) \le 5^i$.
We prove the following lemma:

\begin{lemma}\label{L-bad-approximation}
\renewcommand{\labelenumi}{(\roman{enumi})}
\begin{enumerate}
\item For all $T^h \in \cal{T}$, the LCA of $x$ and $y$ in $T^h$
is in tree level $i - 2$ or higher.

\item There exists at least one $T^h \in \cal{T}$ for which the
LCA of $x$ and $y$ in $T^h$ is in level $i+1$ or lower.
\end{enumerate}
\end{lemma}

\begin{proof}
\noindent (i) Consider an arbitrary tree $T^h \in \cal{T}$, and
nodes $x'$ and $y'$, the respective ancestors of $x$ and $y$ in
$T^h_j$. We have that $d(x,x'),d(y,y') < \frac{25}{24}\cdot 5^j$.
We further have that
$d(x',y')
\ge d(x,y) - d(x,x') -  d(y,y')
> 5^{i-1} - 2\frac{25}{24} \cdot 5^j
= 5^{i-1} - \frac{25}{12} \cdot 5^j$.
Note that for $j < i - 3 - \log_5(25/12)$,
(or equivalently, $j \le i-4$) we have that
$d(x',y') > 2 \cdot 5^{j+2}$, and $x'$ and $y'$ cannot be
siblings. Hence, the LCA of $x$ and $y$ cannot be found in level
$i-3$ or lower and can be found in level $i-2$ or higher.

\smallskip
\noindent (ii) Consider an arbitrary tree $T^h \in \cal{T}$, and
nodes $x'$ and $y'$, the respective ancestors of $x$ and $y$ in
tree level $T^h_j$. Assume that $x'$ was inserted before $y'$. 
There must exist some covering point $z \in Y_{5^{j+2}}$ for which
$d(z,x),d(z,x') < \frac{4}{5} 5^{j+2}$.
If there exist more than one point satisfying this condition, let
$z$ be the first inserted point satisfying the condition. Also recall that
$d(y,y')<\frac{25}{24} \cdot 5^j$.
We have that
$d(z,y')
\le d(z,x) + d(x,y) + d(y,y')
< \frac{4}{5} \cdot 5^{j+2} + 5^i + \frac{25}{24} \cdot 5^j
=\frac{101}{120} \cdot 5^{j+2} + 5^i$.
Now let $T^z$ denote the tree in which $z$ is dominant, and let nodes
$x'$ and $y'$ be the respective ancestors of $x$ and $y$ in $T_z$.
Note that for values $j \ge i-2-\log_5(19/120)$ (or equivalently, for values $j
\ge i$), we have that $d(z,y') \le 5^{j+2}$, and so
$x'$ and $y'$ are both children of $z$ in $T^z$ (or are in fact
the same point). This implies that $x$ and $y$ must be descendants
of $z$. Hence, $x$ and $y$ must have a common ancestor in level
$i+1$ or below.
\end{proof}

The query proceeds by executing an LCA query for $x$ and $y$ in each tree of $\cal{T}$. We
select the lowest node among the ancestors returned from these LCA queries, say $v \in T^h_j$.
By Lemma~\ref{lem:level}, this level is within a constant number of levels of the 
lowest ancestral $6$-neighbors of $x,y$.
Given $v$, the ancestors of $x,y$ in $T_j$ can be located in $2^{O(\lambda)}$ time, and a
$\log c$-jump query then locates nodes within a constant number of levels of the lowest
ancestral $c$-neighbors of $x,y$.

\subsection{Dynamic oracle queries in $O(\log^2 \log n)$ time}\label{sec:backup1}

In this section, we give a dynamic oracle that given $x$ and $y$, finds their lowest 
ancestral $c$-neighbors in the hierarchy, thereby deriving a 
$(1+\eps)$-approximation to $d(x,y)$. We prove the following theorem:

\begin{theorem}\label{thm:backup1}
There exists a dynamic oracle that given $x,y \in S$ returns a 
$(1+\eps)$-approximation to
$d(x,y)$ in $O(\log^2 \log n)$ time, and supports updates in time
$2^{O(\lambda)}\log n + \eps^{-O(\lambda)}$.
\end{theorem}

We begin by presenting a solution for the static version of the problem
and later show how to adapt this solution to the dynamic environment. 
We will make use of the point set $S$ and tree $T$.

\paragraph{Static construction.}
Recall that given $x$ and $y$ it is sufficient to find the lowest
ancestral $c$-neighbors of $x$ and $y$ in order to answer the query. 
This problem could be solved by a simple traversal, in parallel, on the 
paths in $T$ upwards from $x$ and $y$. At each level we check 
whether $a(x)$, the ancestor of $x$, and $a(y)$, the
ancestor of $y$, are $c$-neighbors, and the first encountered $c$-neighbors are
the lowest ancestral $c$-neighbors. (Note though that some
ancestors may not be explicit in certain levels.)
This method may require $\Theta(n)$ time.

To improve this runtime, we note that the hereditary property, 
Property \ref{prop:hereditary} implies that a binary search can
be used. This binary search can be implemented via
level ancestor queries on $T$ (see Section~\ref{sec:prelim}),
and reduces the query time to $O(\log n)$.

To further improve the query time we use a centroid path decomposition $C$
of $T$. A {\em centroid path decomposition} partitions the tree $T$ into a
collection of {\em centroid paths} in the following way. The size of a
node $u$ ($s(u)$) is the number of nodes in the subtree rooted at $u$. Each
centroid path has an associated power of 2, say $2^i$, and all nodes on
the path have size $2^i \le s(u)< 2^{i+1}$. A node $u$ is on the same
centroid path as its parent if their sizes are both between $2^i$ and
$2^{i+1}$ for some $i$.

A well-known property of centroid path decompositions is
that for any node $u$, the path from $u$ to the tree root
traverses at most $\log n$ centroid-paths (along their prefixes). To
utilize this we create a centroid path tree that contains a node for each
centroid path. The centroid path tree has an edge from centroid-path-node $p$ to
centroid-path-node $p'$ if $u$, the head of the path $p'$, is a child (in
$T$) of a node on $p$. It follows from the path-decomposition property
that the height of the centroid-path tree is $O(\log n)$.

To speed up the queries we first perform a binary search along the path from $x$-to-root 
considering only the $O(\log n)$ heads of the centroid paths on the $x$-to-root path. 
This is done by using the centroid path tree and level ancestor queries on the centroid 
path tree. The nodes evaluated are compared with to counterparts (in the same level) in 
the $y$-to-root path in $T$, to see if they are $c$-neighbors. The node on the 
$y$-to-root path on the appropriate level can be found using a level ancestor query (in 
the tree $T$). This search determines which pair of centroid paths (one overlapping the 
path of $x$-to-root and one overlapping the path of $y$-to-root) contains the nodes that 
constitute the lowest ancestral $c$-neighbors. However, these paths themselves may be of 
size $O(n)$. Therefore, we preprocess the following information: We create a {\em 
centroid path graph} with the same node-set as the centroid path tree and an edge 
between two centroid-path-nodes if their paths contain any nodes that are $c$-neighbors. 
The edges are weighted with the lowest level on which there exist $c$-neighbors on these 
paths. Trivially, the centroid path graph is not larger than $T$, and can be
preprocessed in the same time. Once this graph exists, the extraction of the 
lowest ancestral $c$-neighbors is immediate.

\paragraph{Static query time.}
The time to binary search the centroid path tree is $O(\log
\log n)$ as the height of any path (in the centroid path tree) is $O(\log
n)$. Note that although we binary search on both paths, these searches are
done one after the other and, therefore, the time is still $O(\log \log
n)$. Once the two centroid paths that contain the lowest ancestral
$c$-neighbors are found then we in $O(1)$ we can obtain the lowest
ancestral $c$-neighbors because of the preprocessing.

\paragraph{Dynamic construction.}
Now consider the dynamic version of the query problem. A dynamic version
of the above search encounters the following problems
(1) level ancestor queries are not supported in this setting, and
(2) the centroid paths, centroid path tree and graph must be maintained.

Recall that the level ancestor query was used twice, upon $T$ and upon
the centroid path tree. We will show how to remove the query on $T$ and how to
circumvent the level ancestor query upon the centroid path tree.

First, we consider the problem of a dynamic centroid path decomposition.
We will use the method from \cite{CoGo06,CoHa05,KoLe07}.
The general idea of the method is a lazy approach
achieved by changing the size constraints of the centroid paths to have
nodes of size between $2^i$ and $3\cdot 2^{i+1}$. This gives the necessary
time to (lazily) update the centroid path decomposition with worst case
$O(1)$ time per change.

Consider the centroid path tree. Define a directed edge from a leaf to an
ancestor to be an {\em ancestor edge}. We change the centroid path as
follows. The node set, i.e. a node per centroid path, remains the same.
However, the edge set is changed to be the collection of all ancestor
edges. We note that it follows directly from the lazy approach method for
the centroid paths that maintaining the ancestor edges under the dynamic
changes is possible with the same lazy approach. Hence each update can be
implemented in $O(1)$ time. Unfortunately, the number of edges in the
centroid path tree blows up to $O(n\log n)$ instead of the original $n$.
However, this can be corrected by binarizing the tree $T$ and using
indirection on the tree in a method described in \cite{CoHa05,KoLe07}. The
idea follows along the following lines.

The tree $T$ is partitioned into a collection of trees $CT$ of size
$O(\log n)$ such that every node of $T$ is in $CT$ and an edge of $T$ is
in $CT$ if it connects two nodes in the same tree in $CT$. The property of
this partition is that each tree in $CT$ has at most two other children
trees of $CT$. A skeleton tree $\hat{T}$ containing the roots of the
$CT$-trees as nodes and children-parent edges according to the $CT$ tree
relationship are created. See \cite[Section 6]{KoLe07}, for
details of this skeleton tree and its dynamic handling. Obviously the
size of the skeleton tree is $O(n/\log n)$. We will use a centroid path
decomposition on the skeleton tree and create accordingly a centroid path
tree. The centroid path tree can now handle the dynamic changes and
searches and maintain a size of $O(n)$.

A change needs to be made to the
centroid path graph as well. Note that the centroid path graph, as opposed to
the centroid path decomposition and centroid path tree, is unique to this problem.
Beforehand, two centroid paths had an edge
between them if there was a $c$-neighbor pair. We slightly change this
definition such that two nodes (both in the skeleton tree) will be {\em
$c$-pseudo-neighbors} if one of them is a $c$-neighbor of a node in the $CT$
tree of the other. In the centroid path graph two centroid paths will be
neighbors if there are a pair of nodes that are $c$-pseudo-neighbors. The
weight of the edge, similar to before, will be the level of the lowest
level for which we have a pair of $c$-pseudo-neighbors (the level is defined according to the node
with the lower level).

Finally, we need to replace the level ancestor query which we used
upon $T$. This query was done when we had an ancestor of $x$ which was the
head of a centroid path $p$ on some level, say $j$, and we needed to find
its counterpart, i.e. the ancestor of $y$ on level $j$, to see if they are
$c$-neighbors. The replacement will be a binary search on the path from
$y$ to root in $T$ along the heads of the centroid paths. This is done
until we are in the position where we have two centroid paths $p'$ and
$p''$ on the $y$-to-root path, where $p''$ is the son of $p'$ in the
centroid path tree and where the level of the head of the path of $p'$ is
$\geq j$ and the level of the head of $p''$ is $< j$. It can be verified
that the counterpart of the ancestor of $x$ is in a $CT$ tree whose root
is on the centroid path $p'$ and hence if the ancestor of $x$ and it's counterpart are $c$-neighbors then
the ancestor of $x$ and the root of the $CT$ tree (containing the counterpart) are $c$-pseudo-neighbors.
Hence, there is an edge $(p,p')$ in the centroid path graph. Conversely,
if there is an edge $(p,p')$ because the level of the head of $p'$ is
lower than the head of $p$ it follows from the hereditary property that
the mentioned ancestor of $x$ and its counterpart must be $c$-neighbors.
Finding the lowest ancestral $c$-neighbors is done by finding the lowest
pair of nodes (which are $CT$ tree roots that are $c$-pseudo-neighbors). Then
one needs to extract the appropriate node from the $CT$ tree of one which
is on the level of the root of the other. This can be done with a simple
scan in the $CT$-tree.

\paragraph{Dynamic query time.}
A binary search on the path of $x$ can be done in $O(\log\log n)$ as in the static
case. However, for each step in the binary search on the path of $x$, we must
execute a binary search over the path of $y$, in order to locate the ancestor of $y$ in
the correct level. Now, there is the additional step of moving from $c$-pseudo-neighbors to
$c$-neighbors in order to find the lowest ancestral $c$-neighbors may cost
$O(\log n)$ time because of the size of the $CT$ tree. However, if we
recurse the above-described method partitioning each of the $CT$ trees
then we will have small-$CT$ trees of size $O(\log \log n)$ and extracting
the appropriate node will take only another $O(\log \log n)$ steps.

\section{Technical contributions}\label{sec:technical}

In this section we present technical constructions utilized by the distance
oracles.

\subsection{Euclidean distance oracle}\label{sec:bit-trick}
The following lemma utilizes atomic word operations to find the exact distance
between (sparse) Euclidean points in $O(1)$ time.

\begin{lemma}\label{lem:bit-dist}
Let $S$ be a dynamic set of $d$-dimensional vectors, where each coordinate is a
$b$-bit number. If $b\cdot d^2 = O(\log n)$, then there exists a
vector representation of points in $S$ that
\begin{enumerate}
\item
Constructs each vector of $O(1)$ words in $O(1)$ time.
\item
Allows the $\ell_2$ distance between any point pair $p,q \in S$ to be computed
in $O(1)$ time.
\end{enumerate}
\end{lemma}
\begin{proof}
Let $p_i$ be the $i$-th coordinate of $d$-dimensional point $p \in S$, and
recall that the $\ell_2$ distance between two points $p,q \in S$ is
defined as
$\|p-q\|
=\sum_{i=0}^{d-1}(p_i-q_i)^2
=\sum_{i=0}^{d-1}p_i^2
-\sum_{i=0}^{d-1}p_iq_i
+\sum_{i=0}^{d-1}q_i^2$.
It suffices to show
that there exists a vector representation for all points $p,q \in S$
that occupies $O(1)$ words per point and
allows the sum $\sum_{i=0}^{d-1}p_iq_i$ to be computed in $O(1)$ time.

Assume without loss of generality that $d$ is a power of 2, and
for the sake of simplicity, assume that $4b \cdot d^2 \le
\lceil \log n \rceil$, so that all operations below can be done on a
single word. We pad each vector with $d$ additional coordinates (each of $b$ bits all set
to $0$), resulting in $2d$-dimensional vectors.

For each point $p \in S$, we create two vectors $u^p$ and $v^p$. Vector
$u^p$ is constructed as follows. Every coordinate of $p$ is
stored at the rightmost position of a range of $r=2b$ bits, with coordinate
$p_i$ stored in bits $[ir,\ldots,(i+1)r-1]$ for all $0 \leq i < 2d$
(numbered from the right end as usual), with all unused bits set to $0$. Vector $v^p$ is constructed as
follows. Every coordinate in $p$ is stored in the
rightmost position in a range of $r'= 2b \cdot d$ bits, with coordinate $p_i$ stored
in bits $[ir',\ldots,(i+1)r'-1]$, with all unused bits set to $0$.

Now take points $p,q \in S$, and compute in $O(1)$ time the
product $w = u^p \times v^q$. Note that for all $i$, $p_iq_i$ is found in
$r$ consecutive bits beginning at position $i(r'+r)$ of $w$. Set to $0$
all bits of $w$ that do not correspond to a product $p_iq_i$, that is
all bits not in the range $[i(r'+r),i(r'+r)+r]$ for all $i$. (This can be done using bitwise AND with a fixed number.) We are left
with vector $w$ that contains exactly one copy of each product $p_iq_i$.
It remains only to sum these entries in $O(1)$ time.

To this end, let $x$ be a vector that has a
$1$ in the $i(r'+r)$-th bit for every $0\leq i < 2d$
and $0$ elsewhere. Let $y = w \times x$. The sum of the entries of $w$ is
found in $2b +\log d$ bits beginning at position $(r'+r)(d-1)$ of $y$.
\end{proof}

\subsection{Dynamic jump tree}\label{sec:jump-tree}

In this section, we will describe a dynamic structure that supports jump queries.
The compressed hierarchy tree $T$ was described in Section~\ref{sec:backbone}. We now describe
$k$-jump queries on the tree $T$.

\begin{definition}
A $k$-jump query  on compressed hierarchy tree $T$ provides two explicit tree
nodes, $u \in Y_{5^l}$ and its ancestor $w \in Y_{5^p}$. Let $m$ be the largest value less
than $p$ which is a multiple of $k$. The query requests the node
$v \in Y_{5^m}$ that is ancestral to $u$; if $v$ is implicit then its lowest
explicit ancestor is requested instead.
\end{definition}

The existence of a dynamic structure supporting jump queries would allow us to descend $T$
via jumps. 

\begin{lemma}\label{lem:jump}
For fixed $k$, a structure that supports $k$-jump queries of hierarchy tree $T$ can
be maintained along with $O(k)$ work per insertion to $T$ and $O(|T|)$ space.
\end{lemma}

Before presenting a proof of Lemma~\ref{lem:jump}, we first need a preliminary lemma
that extends the dynamic LCA structure of Cole and Hariharan \cite{CoHa05}.

\begin{lemma}\label{lem:lca-child}
For any tree $T$, there exists an LCA query structure that supports insertion of
leaves and internal nodes to $T$ in $O(1)$ time, and answers the following query in
$O(1)$ time: given nodes $u,v \in T$, return $w = \lca(u,v) \in T$ as well as the
children $u',v' \in T$ of $w$ that are the respective ancestors of $u$ and $v$.
\end{lemma}

\begin{proof}
Given tree $T$, we create a new tree $T_1$ as follows. Let $r$ be the root of $T$
and let $v_0,...,v_f$ be $r$'s ordered children. The root of the tree $T_1$ is $r$.
$r$'s left child is $v_0$, and for all nodes $1\leq i \le f$ we have that $v_i$
is the right child of $v_{i-1}$. We then recursively build the subtrees of each child
node $v_i$. This tree can be maintained in $O(1)$ time for each update to $T$.
Now consider nodes $u,v \in T$ that have $w = \lca(u,v) \in T$, and consider the nodes
$u',v' \in T$ that are children of $w$ and the respective ancestors of $u,v \in T$.
Assume that $u'$ precedes $v'$ in the ordering of the children of $w \in T$. Then by
construction, an LCA query on $u,v \in T_1$ returns $u' \in T_1$.

It remains to identify $v'$. To this end, we create tree $T_2$ as follows:
The root of the tree $T_2$ is $r$. $r$'s left child is $v_f$, and for
$0\leq i < f$ we have that $v_i$ is the right child of $v_{i-1}$. We then
recursively create the subtrees of $r$'s children $v_i$.
Now consider nodes $u,v,w,u',v' \in T$ mentioned above. By construction, an LCA query
on $u,v \in T_2$ returns $v' \in T_2$.
\end{proof}

We can now proceed in the proof of Lemma~\ref{lem:jump}.

\begin{proof}
Let tree $T'$ preserve every $k$-th level of $T$. We build $T'$ from $T$ as follows: Level
$T'_j$ contains a copy of every uncompressed node of level $T_{j \cdot k}$,
$j=0,\ldots,\infty$. Further, $T'_j$ contains a copy of every compressed node $x$ in $T_{j
\cdot k}$ that has its lowest uncompressed ancestor $y$ in some level below $T_{j \cdot
(k+1)}$, and $x$ is given a pointer to $y$. This can easily be done in $O(k)$ time per
tree update. (Note that the compression scheme implies that $u$ is the only descendant of
$v$ in level $T_{j \cdot k}$.) The ancestor-descendant relationship in $T'$ is defined by
the anscestor-descendant relationship in $T$.

Now given a $k$-jump query for nodes $u,w \in T$, we first locate the lowest respective
ancestors $u',w'$ of $u,w$ whose tree level is divisible by $k$. (This information can be
maintained for each node in $O(k)$ time.) The LCA query of Lemma~\ref{lem:lca-child} on
$u',v' \in T'$, where $v'$ is a child of $w' \in T'$ which is not an ancestor of $u' \in
T'$, returns $w'$ as well as the child $u'' \in T'$ of $w' \in T'$. $u'' \in T$ (or if it
is compressed, its lowest uncompressed ancestor) is the desired node.
\end{proof}

\subsection{Dynamic embeddings}\label{sec:dyn-embed}

Here we present two randomized dynamic embeddings for an $n$-point metric space $(S,d)$
with doubling dimension $\lambda$. Both embeddings store $O(n)$ interpoint distances
and each can be maintained in time
$2^{O(\lambda)} \min \{\log n, \log \alpha \} + l^{O(\lambda)}$ per update (where $l\ge 5$
is a parameter specific to each embedding).
\begin{itemize}
\item
The first embedding is into a tree metric, with $l=O(\lambda^2)$. Let $T$ be the target space
of the embedding. Given two points $x,y \in S$, we show that $d_T(x,y) \ge d(x,y)$ (that is,
the embedding is non-contractive), and that $d_T(x,y) \ge [O(\lambda)]^i d(x,y)$ with
probability at most $(4/5\lambda)^{i}$ (for any positive integer $i$).
\item
The second embedding is a snowflake embedding into $\ell_2$, with $l=O(1)$. Let $E$ be the
target space of the embedding. Given two points $x,y \in S$, we show that
$\frac{\|f(x)-f(y)\|_2}{d(x,y)^{1/2}} \le 1$ (that is, the embedding is non-expansive to
the snowflake), and that $\frac{\|f(x)-f(y)\|_2}{d(x,y)^{1/2}} > 2^{-11}/\lambda$
with high probability.
\end{itemize}

Both embeddings are build upon the hierarchy of \cite{CoGo06}, after a new
assignment of parent-child relationships to the hierarchical points.

\paragraph{Parent-child assignment.}
We restrict ourselves to consider each $\lceil \log_5 l \rceil$-th level in
the hierarchy. (For ease of presentation, we
will henceforth assume that $l$ is a power of 5.) With regards to this
{\em restricted hierarchy}, a repeated application of the covering property
gives that every point in level $H_{5^{i \log_5 l}} = H_{l^i}$ is within
distance $\frac{4}{5}l^{i+1}$ of some point in level $H_{l^{i+1}}$, and this
constitutes the covering property for the restricted hierarchy.

Let $x \in H_{l^i}$ be a newly inserted point occurrence in the hierarchy.
As in \cite{AbBaNe08}, we associate with $x$ a radius
$r_x \in [l^i,2l^i]$, where $r_x$ is a random
variable sampled from a truncated exponential density function:
The density function is $f(r) = \frac{\lambda^8}{1-\lambda^{-8}} \rho e^{-\rho r}$
with parameter $\rho = 2 \ln (\lambda^4)/r$ when
$r \in [l^i,2l^i]$, and is $f(r)=0$ elsewhere. (This is the construction presented
in \cite{AbBaNe08} with parameter $\Delta=4l^i$.)
Then $x$ is the parent of all subsequently inserted point occurrences in level
$H_{l^{i-1}}$ within distance $r_x$ of $x$, unless those points are within
the radius $r_y$ of a point $y \in H_{l^i}$ that was inserted before $x$.
This defines the parent-child relationship in the restricted hierarchy.

The hierarchy stores $O(n)$ interpoint distances, and can be maintained in
$2^{O(\lambda)} \min \{\log n, \log \alpha \} + 2^{O(\lambda \log \lambda)}$ update
time.

\subsubsection{Tree embedding}\label{sec:dyn-tree}

Here, we present a dynamic embedding of $S$ into a tree metric. We use
the hierarchy and parent-child relationships delineated above, with
$l=37\lambda^2$. We extract a randomized tree from the hierarchy as follows:
For each point occurrence in the restricted hierarchy, there exists a
single corresponding node in the tree. Hence, the parent-child
relationship among the restricted hierarchical points
immediately defines a parent-child relationship in the corresponding tree,
where an edge connect a parent to its child.
From the randomized tree, we extract a tree metric $d_T(\cdot,\cdot)$ by
assigning a length to each edge: An edge rooted at level $H_{d^i}$ is
assigned length $(4l)^i$. We have the following lemma:

\begin{lemma}\label{lem:tree-embed}
For any two points $x,y \in S$ and positive integer $i$, where $l=37\lambda^2$,
\begin{itemize}
\item
$d_T(x,y) \ge d(x,y)$.
\item
$\Pr[\frac{d_T(x,y)}{d(x,y)} > \frac{32}{5}(4l)^{i+1}] \le (\frac{4}{5\lambda})^{i}$.
\end{itemize}
\end{lemma}

\begin{proof}
Consider any two points $x,y \in S$, or equivalently the corresponding point occurrences
$x,y \in H_0$. We first show that the tree embedding is non-contractive: If $x$ and $y$
have their least common ancestor in level $i>0$ of the tree, then by construction
$d_T(x,y) > 2 \cdot 4^i l^i \ge 8l^i$, while
$d(x,y) \le 2 \sum_{j=1}^i 2 l^j < 8 l^i$.
Hence, the embedding is non-contractive.

Next, we derive a probabilistic upper bound on the expansion of the embedding: Let
$l^{k-1} < d(x,y) \le l^k$. Then the true distance between the hierarchical ancestors
of $x,y$ in level $H_{l^m}$, $m \ge k$, of the restricted hierarchy is less than
$d(x,y) + 2 \frac{8}{5} \sum_{j=0}^m l^j
< l^k + \frac{32}{5} l^m
\le \frac{37}{5} l^{m}$.
By the covering property of the restricted hierarchy, $x$ is covered by some point
$x' \in H_{l^{m+1}}$ for which $d(x,x') \le \frac{4}{5}l^{m+1}$, and so a simple
computation gives $y$'s covering point $y' \in H_{l^{m+1}}$ also falls within the radius
$r_{x'}$ of $x'$:
$d(x',y) \le d(x',x) + d(x,y) \le \frac{4}{5}l^{m+1} + \frac{37}{5} l^m \le l^{m+1}$.
Now, the probability that the respective ancestors of $x$ and $y$ in level $H_{l^m}$
do not share the same parent is bounded by
$4\lambda \frac{\frac{37}{5}l^m}{l^{m+1}} = \frac{4}{5\lambda}$
\cite{AbBaNe08}. Hence, the probability that $x$ and $y$ have their lowest common
ancestor at
level $H_{l^{k+i}}$ is bounded by $(4/5\lambda)^{i}$, in which case
$d_T(x,y)
\le 2 \frac{8}{5} \sum_{j=1}^{k+i} (4l)^j
< \frac{32}{5} (4l)^{k+i}
\le d(x,y) \cdot \frac{32}{5} (4l)^{i+1}$.
\end{proof}

\subsubsection{Snowflake embedding}\label{sec:dyn-snow}

In this section we give a dynamic Assouad style embedding \cite{Assouad83}, in which
for a given metric space $(S,d)$ we embed the snowflake $(S,d^{1/2})$ of the metric
into $\ell_2$ space. Our theorem can be viewed as a dynamic version of the theorem
of \cite{GuKrLe03,AbBaNe08}. (A similar embedding holds for $d^{\beta}$ with $0 <
\beta \le 1$ and for general target space $\ell_p$.) For simplicity we focus on the
probabilistic version of the theorem which bounds the distortion with constant
probability.

\begin{theorem}\label{thm:snowflake}
For any $n$ point metric space $(S,d)$ with doubling dimension $\lambda$, there
exists a non-expansive probabilistic embedding $f:S \mapsto E$,
$E \subset \ell_2^D$, that realizes the snowflake $(S,d^{1/2})$: For every pair
$x,y \in S$:
$$\Pr_{f:S\mapsto E}\left[
\frac{\|(f(x)-f(y)\|_2}{d(x,y)^{1/2}} < 2^{-11}/\lambda
\right] \leq e^{-D/16} .$$
Moreover, this construction can be computed dynamically with storage of $O(n)$ interpoint
distances and $2^{O(\lambda)} \min \{\log \alpha, \log n\}$ update time.
\end{theorem}

Our embedding uses the same hierarchy and parent-child relationship presented above,
with $l=8$. Let $H_{l^i}$-cluster $C_x$ be composed of all descendants of $x \in
H_{l^i}$, and call $x$ the center of this cluster. It follows that each point is
found in $O(\log \alpha)$ clusters, one cluster for each level of the hierarchy. Let
$C(l^i,y)$ denote the $H_{l^i}$-cluster containing $y$.

As usual for the construction of snowflake embeddings, we shall construct the
embedding function $f$ by defining for each integer $1 \leq t \leq D$ a function
$f^{(t)}:X \rightarrow \mathbb{R}^{+}$, and then letting
$f = D^{-1/2} \bigoplus_{1 \leq t \leq D} f^{(t)}$.
Fix $t$, $1 \leq t \leq D$, and in what follows we will define $f^{(t)}$:
For each restricted hierarchical level $H_{l^i}$ we define a
function $f^{(t)}_i:S \rightarrow \mathbb{R}^{+}$, and for each point
$x \in S$, let $f^{(t)}(x) = \sum_i f^{(t)}_i(x)$.
Let $\{\sigma^{(t)}_i(C_x) | x \in H_{l^i} \}$ be i.i.d.\ symmetric
$\{0,1\}$-valued Bernoulli random variables.
Let $\tau = \ln 2/ (8 \lambda) \ge 2^{-4}/\lambda$.
The embedding is defined as follows: for each $x \in S$,
\begin{itemize}
\item For each $i$, let $f^{(t)}_i(x) =
\sigma^{(t)}_i(C(l^i,x)) \cdot l^{-i/2}\min\{\tau^{-1} \cdot
g_i(x),l^i\}$,
\end{itemize}
where $g_i(x)$ is a function which computes the distance from $x$
to the boundary of $C(l^i,x)$. This can be computed as follows. Let
$v$ be the center of $C(l^i,x)$ and let $U$ be the set of
$H_{l^i}$-cluster centers within distance $4l^i$ of $v$ which were
inserted into $S$ before the insertion of $v$. For $u$ in $U$
let $r_i(u)$ denote its associated radius. Then:
\begin{itemize}
\item $g_i(x) = \min \{ r_i(v) - d(v,x) , \min_{u \in U} (d(u,x) - r_i(u))
\}$.
\end{itemize}
The function $g_i(x)$ replaces the
expression $d(x,X\setminus P_i(x))$ used in embedding of
\cite{AbBaNe08}. (Note that $g_i(x)$ is not affected by the insertion of new points
into the hierarchy, and can be evaluated in time $2^{O(\lambda)}$.) The following
properties are needed to show that it can be replaced in their analysis:

\begin{claim}
\label{claim:distance-to-boundary}
For every $x,y \in X$:

\begin{itemize}
\item If $C(l^i,x) = C(l^i,y)$ then $|g_i(x) - g_i(y)| \leq d(x,y)$.
\item If $C(l^i,x) \neq C(l^i,y)$ then $\max \{ g_i(x),g_i(y) \} \leq d(x,y)$.
\item $g_i(x) \geq \rho$ with constant probability.
\end{itemize}
\end{claim}
\begin{proof}
\begin{itemize}
\item To prove the first claim is clear from assume that $g_i(y)$ is
minimized for some $u \in U$, then: $g_i(x) - g_i(y) \leq (d(u,x)
- r_i(u)) - (d(u,y) - r_i(u)) \leq d(x,y)$. If $g_i(y)$ is
minimized for $v$ a similar argument applies. Similarly, $g_i(y) -
g_i(x) \leq d(x,y)$.

\item We prove the second claim by contrary assumption that $d(x,y)
< g_i(x)$. It follows that $d(x,y) < r_i(v) - d(v,x)$ which
implies that $d(v,y) \leq d(v,x)+d(x,y) < r_i(v)$. Also for each
$u \in U$, we have $d(x,y) \leq d(u,x) - r_i(u)$ which implies
that $r_i(u) \leq d(u,x) - d(x,y) \leq d(u,y)$ but together these
inequalities imply that $y \in C(l^i,x)$ which is a contradiction.

\item As a consequence of the analysis of \cite{AbBaNe08},
we have with constant probability that $d(v,x)+ \rho \leq r_i(v)$ and also
for every $u \in U$, $d(u,x) - \rho > r_i(u)$. It follows that
$g_i(x) \geq \rho$ with constant probability.
\end{itemize}
\end{proof}

Given Claim~\ref{claim:distance-to-boundary} the analysis of
\cite{AbBaNe08} implies the following:

\begin{lemma}
\label{lemma:embedding-upper-assouad} For any $(x,y) \in X$ and
$t\in [D]$:
\begin{equation*}
|f^{(t)}(x) - f^{(t)}(y)| \leq  2^{7}\lambda\cdot d(x,y)^{1/2}.
\end{equation*}
\end{lemma}

\begin{lemma}
\label{lemma:embedding-lower-assouad} For any $(x,y) \in X$, with
probability at least $1-e^{D/16}$:
\begin{equation*}
\|f(x) - f(y)\|_p \geq  2^{-4}\cdot d(x,y)^{1/2}.
\end{equation*}
\end{lemma}
\begin{proof}
It follows from the Assouad-type argument that with probability at
least $1/8$:
\begin{equation*}
|f^{(t)}(x) - f^{(t)}(y)| \geq  2^{-3}\cdot d(x,y)^{1/2}.
\end{equation*}
The lemma follows from applying a Chernoff bound.
\end{proof}

The theorem follows from an appropriate scaling of the embedding
so to achieve a contractive embedding with the required
properties.

\paragraph{Acknowledgements.} We thank Richard Cole, Robi Krauthgamer,
Manor Mendel and Michiel Smid for helpful conversations.

\bibliographystyle{plain}
\bibliography{oracle-dd}


\end{document}